\documentclass{article}
\usepackage{tikz}
\usetikzlibrary{calc,arrows,positioning}
\usepackage[all]{xy}
\usepackage{hyperref}

\input{ma-remarks.sty}

\title{Remarks on Matsumoto and Amano's normal form\\
  for single-qubit Clifford+$T$ operators}

\author{Brett Giles and Peter Selinger}

\date{}

\begin{document}
\maketitle

\begin{abstract}
  Matsumoto and Amano (2008) showed that every single-qubit
  Clifford+$T$ operator can be uniquely written of a particular form,
  which we call the {\em Matsumoto-Amano normal form}. In this mostly
  expository paper, we give a detailed and streamlined presentation of
  Matsumoto and Amano's results, simplifying some proofs along the
  way. We also point out some corollaries to Matsumoto and Amano's
  work, including an intrinsic characterization of the Clifford+$T$
  subgroup of $SO(3)$, which also yields an efficient $T$-optimal
  exact single-qubit synthesis algorithm. Interestingly, this also
  gives an alternative proof of Kliuchnikov, Maslov, and Mosca's
  exact synthesis result for the Clifford+$T$ subgroup of $U(2)$. 
\end{abstract}

% ----------------------------------------------------------------------
\section{Introduction}

An important problem in quantum information theory is the
decomposition of arbitrary unitary operators into gates from some
fixed universal set {\cite{Nielsen-Chuang}}. Depending on the operator
to be decomposed, this may either be done exactly or to within some
given accuracy $\epsilon$; the former problem is known as {\em exact
  synthesis} and the latter as {\em approximate synthesis}
{\cite{Kliuchnikov-etal}}. Here, we focus on the problem of exact
synthesis for single-qubit operators, using the Clifford+$T$ universal
gate set. Recall that the 192-element Clifford group for one qubit is
generated by the Hadamard gate $H$, the phase gate $S$, and the scalar
$\omega=e^{i\pi/4}$. It is well-known that one obtains a universal
gate set by adding the non-Clifford operator $T$ {\cite{Nielsen-Chuang}}.
\begin{equation}\label{eqn-generators}
  H = \frac{1}{\sqrt{2}}\zmatrix{cc}{1&1\\1&-1}, \quad
  S = \zmatrix{cc}{1&0\\0&i}, \quad
  T = \zmatrix{cc}{1&0\\0&e^{i\pi/4}}, \quad
  \omega=e^{i\pi/4}.
\end{equation}
Matsumoto and Amano {\cite{MA08}} showed that every single-qubit
Clifford+$T$ operator can be uniquely written as a circuit in the following form,
which we call the {\em Matsumoto-Amano normal form:}
\begin{equation}\label{eqn-ma}
 (T\mid\emptyseq)\,(HT\mid SHT)^*\,{\cC}.
\end{equation}
Here, we have used the notation of {\em regular expressions} to denote
a set of sequences of operators; see {\cite{regexp}} for details on
regular expressions. The symbol $\emptyseq$ denotes the empty sequence
of operators, and we used the symbol $\cC$ to denote an arbitrary Clifford
operator.  In words, a Matsumoto-Amano normal form consists of a
rightmost Clifford operator, followed by any number of {\em syllables}
of the form $HT$ or $SHT$, followed by an optional syllable $T$. The most
important properties of the Matsumoto-Amano normal form are:
\begin{itemize}
\item Existence: every single-qubit Clifford+$T$ operator can be
  written in Matsumoto-Amano normal form. Moreover, there is an
  efficient algorithm for converting any operator to normal form.

\item Uniqueness: no operator can be written in Matsumoto-Amano normal
  form in more than one way.

\item $T$-optimality: of all the possible exact decompositions of a
  given operator into the Clifford+$T$ set of gates, the
  Matsumoto-Amano normal form contains the smallest possible number of
  $T$-gates.
\end{itemize}
Despite its enormous usefulness, the Matsumoto-Amano normal form is
still not as widely known as it should be. Matsumoto and Amano's paper
contains a wealth of information that is not readily accessible,
because it is left implicit or only mentioned in proofs, rather than
stated as separate theorems. For example, Matsumoto and Amano's
uniqueness proof implicitly contains an efficient algorithm for
$T$-optimal exact single-qubit synthesis. The concept of {\em
  denominator exponent} is left implicit in the proof of
Theorem~1(II-B), as is the concept of {\em residue}, which appears as
evenness and oddness in properties T1--T9. The normal form's
$T$-optimality is not explicitly stated, although it is an obvious
consequence of the normalization procedure, and is implicitly used in
Section~5.  Also, the correspondence between $T$-count and
denominator exponents is hinted at in equation (15), but not
elaborated upon. Other researchers have later refined these
techniques, for example {\cite{Kliuchnikov-etal}}, {\cite{BS}}, and
more recently {\cite[Section 4]{Gosset-etal}}.

The purpose of this note, which is mostly expository, is to give a
detailed and streamlined presentation of Matsumoto and Amano's
results. In particular, we give a greatly simplified version of the
original uniqueness proof. We explicitly state some facts and
corollaries that were left implicit in Matsumoto and Amano's work.  We
give an intrinsic characterization of the Clifford+$T$ subgroup of
$SO(3)$, which is similar to (and indeed implies) the characterization
of the Clifford+$T$ subgroup of $U(2)$ that was given in
{\cite{Kliuchnikov-etal}}. We also show how to calculate the $T$-count
of an operator of $U(2)$ or $SO(3)$ as a function of its denominator
exponent and residue. Finally, we discuss some alternate normal forms.

% ----------------------------------------------------------------------
\section{Existence}

A single-qubit quantum circuit is just a sequence of operators,
usually taken from some distinguished gate set. In the following, we
often write $A_1A_2\ldots A_n$ for such a circuit consisting of $n$
gates, and it is understood that the gates are applied from right to
left, i.e., as in the notation for matrix multiplication. By slight
abuse of notation, we also use the notation $A_1A_2\ldots A_n$ for the
corresponding operator, i.e., the actual matrix multiplication. It
will always be clear from the context whether we are speaking of a
circuit or its corresponding operator.

\begin{definition}\label{def-sC}
  Let $\sC$ denote the Clifford group on one qubit, generated by $H$,
  $S$, and $\omega$. This group has 192 elements. Let $\sS$ be the
  64-element subgroup generated by $S$, $\omega$, and the Pauli
  operator $X$. Let $\sCp = \sC\sm\sS$. Let
  $\sH = \s{I,H,SH}$ and $\sHp = \s{H,SH}$.
\end{definition}

\begin{lemma}\label{lem-SH}
  The following hold:
  \begin{align}
    \sC ~&=~ \sH\sS,\label{eqn-CHS}\\
    \sCp ~&=~ \sHp\sS,\label{eqn-CSHS}\\
    \sS\sHp ~&\seq~ \sHp\sS,\label{eqn-SHHS}\\
    \sS T ~&=~ T \sS,\label{eqn-STTS}\\
    T \sS T ~&=~ \sS.\label{eqn-TSTS}
  \end{align}
\end{lemma}

\begin{proof}
  Since $\sS$ is a 64-element subgroup of $\sC$, it has three left
  cosets. They are $\sS$, $H\sS$, and $SH\sS$. Since $\sC$ is the
  disjoint union of these cosets, (\ref{eqn-CHS}) and (\ref{eqn-CSHS})
  immediately follow. For (\ref{eqn-SHHS}), first notice that $\sS
  S=\sS$, and therefore $\sS\sHp = \sS H\cup \sS SH = \sS H$. Since
  $\sS H$ is a non-trivial right coset of $\sS$, it follows that $\sS
  H \seq \sC\sm\sS = \sCp$. Combining these facts with
  (\ref{eqn-CSHS}), we have (\ref{eqn-SHHS}). Finally, the equations
  (\ref{eqn-STTS}) and (\ref{eqn-TSTS}) are trivial consequences of
  the equations $ST=TS$, $XT=TXS\omega\inv$, $\omega T=T\omega$, and $TT=S$.
\end{proof}

\begin{theorem}[Matsumoto and Amano {\cite[Thm 1(I)]{MA08}}]\label{thm-ex}
  Every single-qubit Clifford+$T$ operator can be written in
  Matsumoto-Amano normal form.
\end{theorem}

\begin{proof}
  Let $M$ be a single-qubit Clifford+$T$ operator. By definition, $M$ can be
  written as
\begin{equation}\label{eqn-M}
  M = C_n\;T\;C_{n-1}\;\cdots\;C_1\;T\;C_0,
\end{equation}
for some $n\geq 0$, where $C_0,\ldots,C_n\in\sC$. First note that if
$C_i\in\sS$ for any $i\in\s{1,\ldots,n-1}$, then we can immediately
use (\ref{eqn-TSTS}) to replace $TC_iT$ by a single Clifford
operator. This yields a shorter expression of the form (\ref{eqn-M})
for $M$. We may therefore assume without loss of generality that
$C_i\not\in\sS$ for $i=1,\ldots,n-1$. If $n=0$, then $M$ is a Clifford
operator, and there is nothing to show. Otherwise, we have
\begin{align}
  M ~&\in~ \sC\;T\;\sCp\;\cdots\;\sCp\;T\;\sC &&\mbox{ by (\ref{eqn-M})}\\
  &=~ \sH\sS\;T\;\sHp\sS\;\cdots\;\sHp\sS\;T\;\sC &&\mbox{ by
    (\ref{eqn-CHS}) and (\ref{eqn-CSHS})}\\
  &\seq~ \sH\;T\;\sHp\;\cdots\;\sHp\;T\;\sC &&\mbox{ by
    (\ref{eqn-SHHS}) and (\ref{eqn-STTS}).}\label{eqn-HTHpTC}
\end{align}
Note how, in the last step, the relations (\ref{eqn-SHHS}) and
(\ref{eqn-STTS}) were used to move all occurrences of $\sS$ to the
right, where they were absorbed into the final $\sC$.  It is now
trivial to see that every element of (\ref{eqn-HTHpTC}) can be written
in Matsumoto-Amano normal form, finishing the proof.
\end{proof}

\begin{corollary}[Matsumoto and Amano {\cite[p.8]{MA08}}]\label{cor-efficient}
  There exists a linear-time algorithm for symbolically reducing any
  sequence of Clifford+$T$ operators to Matsumoto-Amano normal form.
  More precisely, this algorithm runs in time at most $O(n)$, where
  $n$ is the length of the input sequence.
\end{corollary}

\begin{proof}
  The proof of Theorem~\ref{thm-ex} already contains an algorithm
  for reducing any sequence of Clifford+$T$ operators to
  Matsumoto-Amano normal form. However, in the stated form, it is
  perhaps not obvious that the algorithm runs in linear time. Indeed,
  a naive implementation of the first step would require up to $n$
  searches of the entire sequence for a term of the form $T\sS T$,
  which can take time $O(n^2)$.

  One obtains a linear time algorithm from the following observation:
  if $M$ is already in Matsumoto-Amano normal form, and $A$ is either
  a Clifford operator or $T$, then $MA$ can be reduced to
  Matsumoto-Amano normal form in constant time. This is trivial when
  $A$ is a Clifford operator, because it will simply be absorbed into
  the rightmost Clifford operator of $M$. In the case where $A=T$, a
  simple case distinction shows that at most the rightmost 5 elements
  of $MA$ need to be updated. The normal form of a sequence of
  operators $A_1A_2\ldots A_n$ can now be computed in linear time by
  starting with $M=I$ and repeatedly right-multiplying by
  $A_1,\ldots,A_n$, reducing to normal form after each step.
\end{proof}

%----------------------------------------------------------------------
\section{$T$-Optimality}

\begin{corollary}
  Let $M$ be single-qubit Clifford+$T$ operator, and assume that $M$
  can be written with $T$-count $n$. Then there exists a
  Matsumoto-Amano normal form for $M$ with $T$-count at most $n$.
\end{corollary}

\begin{proof}
  This is an immediate consequence of the proof of
  Theorem~\ref{thm-ex}, because the reduction from (\ref{eqn-M})
  to (\ref{eqn-HTHpTC}) does not increase the $T$-count.
\end{proof}

%----------------------------------------------------------------------
\section{Uniqueness}

\begin{theorem}[Matsumoto and Amano {\cite[Thm 1(II)]{MA08}}]\label{thm-ma}
  If $M$ and $N$ are two different Matsumoto-Amano normal forms, then
  they describe different operators.
\end{theorem}

We give a simplified version of Matsumoto and Amano's proof. Like
Matsumoto and Amano, we use the Bloch sphere representation of unitary
operators.  Recall that each single-qubit unitary operator can be
represented as a rotation of the Bloch sphere, or equivalently, as an
element of $SO(3)$, the real orthogonal $3\times 3$ matrices with
determinant 1. The relationship between an operator $U\in U(2)$ and
its Bloch sphere representation $\hat U\in SO(3)$ is given by
\begin{equation}\label{eqn-bloch-conversion}
  \hat U\zmatrix{c}{x\\y\\z} ~=~ \zmatrix{c}{x'\\y'\\z'}
  \quad\Longleftrightarrow\quad
  U\,(xX+yY+zZ)\;U\da ~=~ x'X + y'Y + z'Z,
\end{equation}
where $X$, $Y$, and $Z$ are the Pauli operators. The Bloch sphere
representations of the operators $H$, $S$, and $T$ are:
\begin{equation}\label{bloch-generators}
 \hat H = \zmatrix{ccc}{0&0&1\\0&-1&0\\1&0&0}, \quad
\hat S = \zmatrix{ccc}{0&-1&0\\1&0&0\\0&0&1}, \quad
\hat T = \frac{1}{\sqrt{2}}\xmatrix{.9}{.8}{ccc}{1&-1&0\\1&1&0\\0&0&\sqrt{2}}.
\end{equation}
The assignment $U\mapsto \hat U$ defines a group homomorphism from
$U(2)$ to $SO(3)$, and we write $\hat\sC$ for the image of $\sC$ under
this homomorphism.

\begin{remark}\label{rem-bloch}
  The elements of $\hat\sC$ are the Bloch sphere representations of
  the Clifford operators. Since global phases are lost in the Bloch
  sphere representation, there are 24 such operators. They are
  precisely those elements of $SO(3)$ that can be written with matrix
  entries in $\s{-1,0,1}$, or equivalently, the 24 symmetries of the
  cube $\s{(x,y,z)\mid -1\leq x,y,z\leq 1}$.
\end{remark}

\begin{definition}
  Recall that $\N$ denotes the natural numbers including 0; $\Z$
  denotes the integers; and $\Zb$ denotes the integers modulo 2.
  We define three subrings of the real numbers:
  \begin{itemize}
  \item $\D = \Z[\frac12] = \s{\frac{a}{2^n}\mid a\in\Z,
      n\in\N}$. This is the ring of {\em dyadic fractions}.
  \item $\Zs = \s{a+b\sqrt 2 \mid a,b\in\Z}$. This is the ring of {\em
      quadratic integers} with radicand 2.
  \item $\Ds = \Z[\frac1{\sqrt2}] = \s{r+s\sqrt 2 \mid r,s\in\D}$.
  \end{itemize}
  We will also need the following two subrings of the complex
  numbers. Recall that $\omega = e^{i\pi/4} = (1+i)/\sqrt{2}$ is an
  $8$th root of unity satisfying $\omega^2=i$ and $\omega^4=-1$.
  \begin{itemize}
  \item $\Zw = \s{a\omega^3+b\omega^2+c\omega+d \mid
      a,b,c,d\in\Z}$. This is the ring of {\em cyclotomic integers of
      degree 8}.
  \item $\Dw = \s{a\omega^3+b\omega^2+c\omega+d \mid a,b,c,d\in\D}$.
  \end{itemize}
\end{definition}

\begin{remark}\label{rem-necessary}
  If $U$ is a Clifford+$T$ operator, then its matrix entries are in
  the ring $\Dw$. This is trivially true, because it holds for each of
  the generators {\eqref{eqn-generators}}. Moreover, the entries of
  the corresponding Bloch sphere operator $\hat U$ are from the ring
  $\Ds$. This is also trivial from {\eqref{bloch-generators}}.
\end{remark}

\begin{definition}[Parity]
  Consider the unique ring homomorphism $\Z\to\Zb$, mapping $a\in\Z$
  to $\parity a\in\Zb$, where $\parity a=0$ if $a$ is even and
  $\parity a=1$ if $a$ is odd. We define the {\em parity map}
  $p:\Zs\to\Zb$ by $p(a+b\sqrt{2}) = \parity a$. Note that this is
  also a ring homomorphism. We refer to $p(x)$ as the {\em parity} of
  $x$.
\end{definition}

\begin{definition}[Denominator exponent]\label{def-denomexp}
  For every element $q\in\Ds$, there exists some natural number $k\geq
  0$ such that $\rt{k}q\in\Zs$, or equivalently, such that $q$ can
  be written as $x/\rt{k}$, for some quadratic integer
  $x$. Such $k$ is called a {\em denominator exponent} for $q$. The
  least such $k$ is called the {\em least denominator exponent} of
  $q$.

  More generally, we say that $k$ is a denominator exponent for a
  vector or matrix if it is a denominator exponent for all of its
  entries. The least denominator exponent for a vector or matrix is
  therefore the least $k$ that is a denominator exponent for all of
  its entries.
\end{definition}

\begin{definition}[$k$-parity]
  Let $k$ be a denominator exponent for $q\in\Ds$. We define the {\em
    $k$-parity} of $q$, in symbols $p_k(q)\in\Zb$, by $p_k(q) =
  p(\rt{k} q)$.  The $k$-parity of a vector or
  matrix is defined componentwise.
\end{definition}

% ......................................................................
\begin{figure}
  \[
  \xymatrix{
    *+{\makebox[0in][r]{Start:~}\zmatrix{ccc}{1&0&0 \\ 0&1&0 \\ 0&0&1}}
    \ar[d]^<>(.5){\displaystyle T}_<>(.5){k\pp}
    \ar@<1.5ex>@(dr,r)[]+R_<>(.53){\displaystyle \sC}
    \\
    *+{\zmatrix{ccc}{1&1&0 \\ 1&1&0 \\ 0&0&0}}
    \ar@<-0.5ex>@/_/[d]_<>(.5){\displaystyle H}
    \\
    *+{\zmatrix{ccc}{0&0&0 \\ 1&1&0 \\ 1&1&0}}
    \ar[r]_<>(.5){\displaystyle S}
    \ar@<-0.5ex>@/_/[u]_<>(.6){\displaystyle T}_<>(.25){k\pp}
    &
    *+{\zmatrix{ccc}{1&1&0 \\ 0&0&0 \\ 1&1&0}}
    \ar@<-2ex>@/_4ex/[ul]_<>(.5){\displaystyle T}^<>(.5){k\pp}
  }
  \]
  \caption{The left action of Matsumoto-Amano normal forms on
    $k$-parities over $SO(3)$. All matrices are written modulo the
    right action of the Clifford group, i.e., modulo a permutation of
    the columns.}\label{fig-so3-action}
\end{figure}
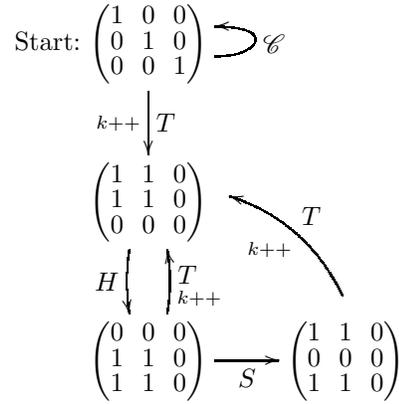
% ......................................................................

\begin{remark}\label{rem-right-action}
  Let $C$ be any Clifford operator, and $\hat C$ its Bloch sphere
  representation. As noted above, the matrix entries of $\hat C$ are
  in $\s{-1,0,1}$; it follows that $\hat C$ has denominator exponent
  0. In particular, it follows that multiplication by $\hat C$ is a
  well-defined operation on parity matrices: for any $3\times
  3$-matrix $U$ with entries in $\Zb$, we define $U\bullet \hat C :=
  U\cdot p(\hat C)$. This defines a right action of the Clifford group
  $\hat\sC$ on the set of parity matrices.
\end{remark}

\begin{definition}\label{def-simg}
  Let $\sim_{\hat\sC}$ the be the equivalence relation induced by this
  right action. In other words, for parity matrices $U,V$, we write
  $U\sim_{\hat\sC} V$ if there exists some $\hat C\in \hat\sC$ such
  that $U\bullet \hat C=V$. In elementary terms, $U\sim_{\hat\sC} V$
  holds if and only if $U$ and $V$ differ by a permutation of columns.
\end{definition}

\begin{lemma}\label{lem-ma}
  Let $M$ be a Matsumoto-Amano normal form, and $\hat M\in SO(3)$ the Bloch
  sphere operator of $M$.  Let $k$ be the least denominator exponent
  of $\hat M$. Then exactly one of the following holds:\rm
  \begin{itemize}
  \item $k=0$, and $M$ is a Clifford operator.
  \item $k>0$, $p_k(\hat M)\sim_{\hat\sC} \zmatrix{ccc}{1&1&0 \\ 1&1&0 \\
      0&0&0}$, and the leftmost syllable of $M$ is $T$.
  \item $k>0$, $p_k(\hat M)\sim_{\hat\sC} \zmatrix{ccc}{0&0&0 \\ 1&1&0 \\ 1&1&0}$, and the
    leftmost syllable of $M$ is $HT$.
  \item $k>0$, $p_k(\hat M)\sim_{\hat\sC} \zmatrix{ccc}{1&1&0 \\ 0&0&0 \\
      1&1&0}$, and the leftmost syllable of $M$ is $SHT$.
  \end{itemize}
  Moreover, the $T$-count of $M$ is equal to $k$.
\end{lemma}

\begin{proof}
  By induction on the length of the Matsumoto-Amano normal form $M$.
  Figure~\ref{fig-so3-action} shows the action of Matsumoto-Amano
  operators on parity matrices. Each vertex represents a
  $\sim_{\hat\sC}$-equivalence class of $k$-parities. The vertex
  labelled ``Start'' represents the empty Matsumoto-Amano normal form,
  i.e., the identity operator. Each arrow represents left
  multiplication by the relevant operator, i.e., a Clifford operator,
  $T$, $H$, or $S$. Thus, each Matsumoto-Amano normal form, read from
  right to left, gives rise to a unique path in the graph of
  Figure~\ref{fig-so3-action}. The label $k\pp$ on an arrow indicates
  that the least denominator exponent increases by $1$. An easy case
  distinction shows that the parities and least denominator exponents
  indeed behave as shown in Figure~\ref{fig-so3-action}.  The claims
  of the lemma then immediately follow.
\end{proof}

\begin{proof}[Proof of Theorem~\ref{thm-ma}]
  This is an immediate consequence of Lemma~\ref{lem-ma}. Indeed,
  suppose that $M$ and $N$ are two Matsumoto-Amano normal forms
  describing the same unitary operator $U$. We show that $M=N$ by
  induction on the length of $M$. Let $k$ be the least denominator
  exponent of $U$. If $k=0$, then by Lemma~\ref{lem-ma}, both $M$ and
  $N$ are Clifford operators; they are then equal by assumption. If
  $k>0$, then by Lemma~\ref{lem-ma}, the Matsumoto-Amano normal forms
  $M$ and $N$ have the same leftmost syllable (either $T$, $HT$, or
  $SHT$), and the claim follows by induction hypothesis.
\end{proof}

%----------------------------------------------------------------------
\section{The Matsumoto-Amano exact synthesis algorithm}

As an immediate consequence of Lemma~\ref{lem-ma}, we obtain an
efficient algorithm for calculating the Matsumoto-Amano normal form of
any Clifford+$T$ operator, given as a matrix.

\begin{theorem}\label{thm:ma-decomposition}
  Let $U\in U(2)$ be some Clifford+$T$ operator. Let $k$ be the least
  denominator exponent of its Bloch sphere representation $\hat U$.
  Then the Matsumoto-Amano normal form $M$ of $U$ can be efficiently
  computed with $O(k)$ arithmetic operations.
\end{theorem}

\begin{proof}
  Given $U$, first compute its Bloch sphere representation $\hat
  U$. This requires only a constant number of arithmetic operations by
  {\eqref{eqn-bloch-conversion}}. Let $k$ be the least denominator
  exponent of $\hat U$.  Let $M$ be the unique (but as yet unknown)
  Matsumoto-Amano normal form of $U$. Note that, by
  Lemma~\ref{lem-ma}, the $T$-count of $M$ is $k$. We compute $M$
  recursively. If $k=0$, then $M$ is a Clifford operator by
  Lemma~\ref{lem-ma}, and we have $M=U$. If $k>0$, we compute
  $p_k(U)$, which must be of one of the three forms listed
  in Lemma~\ref{lem-ma}. This determines whether the leftmost syllable
  of $M$ is $T$, $HT$, or $SHT$. Let $N$ be this syllable, so that
  $M=NM'$, for some Matsumoto-Amano normal form $M'$. Then $M'$ can be
  recursively computed as the Matsumoto-Amano normal form of $U'=
  N\inv U$; moreover, since $M'$ has $T$-count $k-1$, the recursion
  terminates after $k$ steps. Since each recursive step only requires
  a constant number of arithmetic operations, the total number of
  operations is $O(k)$.
\end{proof}

%----------------------------------------------------------------------
\section{A characterization of Clifford+$T$ on the Bloch sphere}

Theorem~\ref{thm:ma-decomposition} states that {\em if} $U$ is a
Clifford+$T$ operator, then an actual Clifford+$T$ circuit for it can
be efficiently synthesized. Trivially, this also yields a method for
{\em checking} whether a given operator $U$ with entries in the ring
$\Dw$ is in the Clifford+$T$ group: namely, apply the algorithm of
Theorem~\ref{thm:ma-decomposition}. This either yields a Clifford+$T$
decomposition of $U$, or else the algorithm fails. The algorithm could
potentially fail in three different ways: (a) at some step, $p_k(U)$
is not of one of the three forms listed in Lemma~\ref{lem-ma}; (b) at
some step, $k$ fails to decrease; or (c) we reach $k=0$ but the
operator $U$ is not Clifford.

Remarkably, none of these three failure conditions can ever happen:
Provided that $U$ is unitary with entries from $\Dw$, the algorithm of
Theorem~\ref{thm:ma-decomposition} will {\em always} yield a
Clifford+$T$ decomposition of $U$. This yields a kind of converse to
the first part of Remark~\ref{rem-necessary}, and a nice algebraic
characterization of the Clifford+$T$ group: it is exactly the group of
unitary matrices over the ring $\Dw$. This result was first proved by
Kliuchnikov et al.~{\cite{Kliuchnikov-etal}}, and later generalized to
multi-qubit operators in {\cite{Giles-Selinger}}.

We now show that Matsumoto and Amano's method also yields a converse
to the second part of Remark~\ref{rem-necessary}: an element of
$SO(3)$ is the Bloch sphere representation of some Clifford+$T$
operator if and only if its matrix entries are from the ring $\Ds$.
This is the Bloch sphere analogue of the theorem of
{\cite{Kliuchnikov-etal}}. Remarkably, the Bloch sphere version of
this result is actually stronger than the $U(2)$ version.

\begin{lemma}\label{lem:dot_product_is_simpler_in_ds}
  Let $U\in SO(3)$ be a special orthogonal matrix with entries in $\Ds$.
  Let $k$ be a denominator exponent of $U$, and let $v_1,v_2,v_3$ be
  the columns of $U$, with
  \[v_{\jay} = \frac{1}{\rt{k}}\left(\begin{matrix}
      a_{\jay}+b_{\jay}\sqrt{2} \\
      c_{\jay}+d_{\jay}\sqrt{2} \\
      e_{\jay}+f_{\jay}\sqrt{2}
    \end{matrix}\right),\]
  for $a_{\jay},\ldots,f_{\jay}\in\Z$.
  Then for all $\jay,\ell\in\s{1,2,3}$,
  \begin{equation}\label{eqn-abba}
    a_{\jay} b_{\ell}+b_{\jay} a_{\ell}+
            c_{\jay} d_{\ell}+d_{\jay} c_{\ell}+e_{\jay} f_{\ell}+f_{\jay} e_{\ell} = 0
  \end{equation}
  and
  \begin{equation}\label{eqn-aacc}
    a_{\jay} a_{\ell}+c_{\jay} c_{\ell}+e_{\jay} e_{\ell}+
    2(b_{\jay} b_{\ell}+d_{\jay} d_{\ell}+f_{\jay} f_{\ell})
    = 2^k\delta_{\jay,\ell}.
  \end{equation}
  Here $\delta_{\jay,\ell}$ denotes the Kronecker delta function. In
  particular, we have, for all $\jay\in\s{1,2,3}$,
  \begin{equation}\label{eqn-abcd}
    a_{\jay} b_{\jay}+
    c_{\jay} d_{\jay}+e_{\jay} f_{\jay} = 0
  \end{equation}
  and
  \begin{equation}\label{eqn-a2c2}
    a_{\jay}^2 + c_{\jay}^2 + e_{\jay}^2 + 2(b_{\jay}^2 + d_{\jay}^2 + f_{\jay}^2) = 2^k.
  \end{equation}
\end{lemma}

\begin{proof}
  Computing the inner product, we have
  \begin{equation}  \label{eqn-dotprod_of_vl_and_vj}
    \p{v_{\jay},v_{\ell}} =
    \frac{1}{2^k}
    \left(
      a_{\jay} a_{\ell}+c_{\jay} c_{\ell}+e_{\jay} e_{\ell}+
      2(b_{\jay} b_{\ell}+d_{\jay} d_{\ell}+f_{\jay} f_{\ell})
      +\sqrt{2}(a_{\jay} b_{\ell}+b_{\jay} a_{\ell}+
      c_{\jay} d_{\ell}+d_{\jay} c_{\ell}+e_{\jay} f_{\ell}+f_{\jay} e_{\ell})
    \right).
  \end{equation}
  Since $U$ is orthogonal, we have $\p{v_{\jay},v_{\jay}}=1$, and
  $\p{v_{\jay},v_{\ell}} = 0$ when $\ell \ne \jay$.
  Therefore, the coefficient of $\sqrt{2}$ in equation
  \eqref{eqn-dotprod_of_vl_and_vj} must be zero, proving
  \eqref{eqn-abba} and \eqref{eqn-aacc}. Equations {\eqref{eqn-abcd}} and
  {\eqref{eqn-a2c2}} immediately follow by letting $\jay=\ell$.
\end{proof}

\begin{remark}
  In Lemma~\ref{lem:dot_product_is_simpler_in_ds}, we have worked with
  columns $v_{\jay}$ of the matrix $U$. But since $U$ is orthogonal, the
  analogous properties also hold for the rows of $U$.
\end{remark}

\begin{lemma}\label{lem:k0-clifford-new}
  Let $U\in SO(3)$ be a special orthogonal matrix with entries in $\Ds$, and
  with least denominator exponent $k$. If $k=0$, then $U$ the Bloch sphere
  representation of some Clifford operator. If $k>0$, then $p_k(U)\sim_{\hat\sC} M$ for some
  $M\in\s{M_T,M_H,M_S}$, where
  \[ M_T =
  \zmatrix{ccc}{1&1&0 \\ 1&1&0 \\ 0&0&0},\quad M_H =
  \zmatrix{ccc}{0&0&0 \\ 1&1&0 \\ 1&1&0},\quad M_S =
  \zmatrix{ccc}{1&1&0 \\ 0&0&0 \\ 1&1&0}.
   \]
\end{lemma}

\begin{proof}
  First consider the case $k=0$. Let $v_{\jay}$ be any column of $U$, with the notation of
  Lemma~\ref{lem:dot_product_is_simpler_in_ds}. By {\eqref{eqn-a2c2}},
  we have
  $a_{\jay}^2+c_{\jay}^2+e_{\jay}^2+2(b_{\jay}^2+d_{\jay}^2+f_{\jay}^2)=1$. Since
  each summand is a positive integer, we must have
  $b_{\jay},d_{\jay},f_{\jay} = 0$, and exactly one of $a_{\jay}$, $c_{\jay}$ or
  $e_{\jay}=\pm1$, for each $\jay=1,2,3$. Therefore, all the matrix
  entries are in $\s{-1,0,1}$, and the claim follows by
  Remark~\ref{rem-bloch}.

  Now consider the case $k>0$. Let $v_{\jay}$ be any row or column of
  $U$, with the notation of
  Lemma~\ref{lem:dot_product_is_simpler_in_ds}. By {\eqref{eqn-a2c2}},
  it follows that $a_{\jay}^2 + c_{\jay}^2 + e_{\jay}^2$ is even, and
  therefore an even number of $a_{\jay}$, $c_{\jay}$, and $e_{\jay}$
  have parity $1$. Therefore, each row or column of $p_k(U)$ has an
  even number of $1$'s. Moreover, since $k$ is the least denominator
  exponent of $U$, $p_k(U)$ has at least one non-zero entry. Modulo a
  permutation of columns, this leaves exactly four possibilities for
  $p_k(U)$:
  \[
  (a)~ \zmatrix{ccc}{1&1&0 \\ 1&1&0 \\ 0&0&0},\quad
  (b)~ \zmatrix{ccc}{0&0&0 \\ 1&1&0 \\ 1&1&0},\quad
  (c)~ \zmatrix{ccc}{1&1&0 \\ 0&0&0 \\ 1&1&0},\quad
  (d)~ \zmatrix{ccc}{1&1&0 \\ 1&0&1 \\ 0&1&1}.
   \]
   In cases (a)--(c), we are done. Case (d) is impossible because
   it implies that $a_1a_2+c_1c_2+e_1e_2$ is odd, contradicting the
   fact that it is even by {\eqref{eqn-aacc}}.
\end{proof}

\begin{lemma}\label{lem:applying_a_t_symbol_decreases_k}
  Let $U\in SO(3)$ be a special orthogonal matrix with entries in $\Ds$, and
  with least denominator exponent $k > 0$. Then there exists
  $N\in\s{T, HT, SHT}$ such that the least denominator exponent of
  $\hat N\inv U$ is $k-1$.
\end{lemma}

\begin{proof}
  By Lemma~\ref{lem:k0-clifford-new}, we know that
  $p_k(U)\sim_{\hat\sC} M$, for some $M\in\{M_T, M_H, M_S\}$. We consider each
  of these cases.
  \begin{enumerate}
    \item \label{lemparitycase:applyt}$p_k(U) \sim_{\hat\sC} M_T$.
      By assumption, $U$ has two columns $v$ with $p_k(v) = (1,1,0)^T$.
      Let
      \[ v =
      \frac{1}{\rt{k}}\zmatrix{c}{a+b\sqrt{2}\\c+d\sqrt{2}\\e+f\sqrt{2}}
      \]
      be any such column. By {\eqref{eqn-abcd}}, we have $ab+cd+ef =
      0$. Since $\parity e = 0$, we have $\parity a \parity b+ \parity
      c \parity d = 0$. Since $\parity a = \parity c = 1$, we can
      conclude $\parity b+\parity d = 0$.  Applying $\hat{T}^{-1}$ to
      $v$, we compute:
      \[\hat{T}^{-1}v = \frac{1}{\rt{k+1}}
      \arraycolsep=.5ex
      \begin{pmatrix}
        c+a &+& (d+b)\sqrt{2} \\
        c-a &+& (d-b)\sqrt{2} \\
        e\sqrt{2} &+& 2f
      \end{pmatrix}
      = \frac{1}{\rt{k-1}}
      \begin{pmatrix}
        \frac{c+a}{2} &+& \frac{d+b}{\sqrt{2}} \\
        \frac{c-a}{2} &+& \frac{d-b}{\sqrt{2}} \\
        \frac{e}{\sqrt{2}} &+& f
      \end{pmatrix}
      = \frac{1}{\rt{k-1}}
      \begin{pmatrix}
        a' &+& b'{\sqrt{2}} \\
        c' &+& d'{\sqrt{2}} \\
        f &+ & e'{\sqrt{2}}
      \end{pmatrix}
      \]
      where $a' = \frac{c+a}{2}, b' = \frac{d+b}{2}, c' =
      \frac{c-a}{2}, d' = \frac{d-b}{2}$ and $e' = \frac{e}{2}$ are
      all integers. Hence, $k-1$ is a denominator exponent of
      $\hat{T}^{-1}v$. Moreover, since $a'+c'=c$ is odd, one of $a'$
      and $c'$ is odd, proving that $k-1$ is the least denominator
      exponent of $\hat{T}^{-1}v$.

      Now consider the third column $w$ of $U$, where $p_k(w) =
      (0,0,0)^T$. Then $k-1$ is a denominator exponent for $w$, so that
      $k$ is a denominator exponent for $\hat{T}^{-1}w$. Let
      \[ p_k(\hat{T}^{-1}w) = \zmatrix{c}{x\\y\\z}.
      \]
      As the least denominator exponent of the other two columns of
      $p_k(\hat{T}^{-1}U)$ is $k-1$, we have
      \[p_k(\hat{T}^{-1}U) \sim_{\hat\sC}
      \begin{pmatrix} 0&0&x \\0&0&y \\0&0&z\end{pmatrix}.
      \]
      But $\hat{T}^{-1}U$ is orthogonal, so by {\eqref{eqn-a2c2}},
      applied to each row of $\hat{T}^{-1}U$, we conclude that
      $x=y=z=0$. It follows that the least denominator exponent of
      $\hat{T}^{-1}U$ is $k-1$.
    \item \label{lemparitycase:applyh}$p_k(U) \sim_{\hat\sC} M_H$.
      In this case, $p_k(\hat{H}^{-1}U) \sim_{\hat\sC} p(\hat{H}^{-1})M_H = M_T$.
      We then continue as in case {\ref{lemparitycase:applyt}}.
    \item \label{lemparitycase:applys}$p_k(U) \sim_{\hat\sC} M_S$.
      In this case, $p_k(\hat{H}\inv\hat{S}^{-1}U) \sim_{\hat\sC} p(\hat{H}\inv\hat{S}^{-1})M_S = M_T$.
      We then continue as in case {\ref{lemparitycase:applyt}}.\qedhere
  \end{enumerate}
\end{proof}

Combining Lemmas~\ref{lem:k0-clifford-new} and
{\ref{lem:applying_a_t_symbol_decreases_k}}, we easily get the
following result:

\begin{theorem}\label{thm:completness_of_algorithm}
  Let $U\in SO(3)$. Then $U$ is the Bloch sphere representation of
  some Clifford+$T$ operator if and only if the entries of $U$ are in
  the ring $\Ds$. Moreover, a Matsumoto-Amano normal form this
  operator can be efficiently computed.
\end{theorem}

\begin{proof}
  The ``only if'' direction is trivial by Remark~\ref{rem-necessary}.
  To prove the ``if'' direction, let $k$ be the least denominator
  exponent of $U$. We proceed by induction on $k$. If $k=0$, by
  Lemma~\ref{lem:k0-clifford-new}, $U$ is the Bloch sphere representation
  of some Clifford operator, and therefore of a Clifford+$T$ operator. If
  $k>0$, then by
  Lemma~\ref{lem:applying_a_t_symbol_decreases_k}, we can write
  $U=\hat NU'$, where $N\in\s{T,HT,SHT}$ and $U'$ has least
  denominator exponent $k-1$. By induction hypothesis, $U'$ is a
  Clifford+$T$ operator, and therefore so is $U$.
\end{proof}

\begin{remark}
  Combining this result with the algorithm of
  Theorem~\ref{thm:ma-decomposition}, we have a linear-time algorithm
  for computing a Matsumoto-Amano normal form for any Bloch sphere operator
  $U\in SO(3)$ with entries in $\Ds$. This normal form will be unique
  up to a global phase.
\end{remark}

As a corollary, we also get a new proof of the following result by
Kliuchnikov et al.~{\cite{Kliuchnikov-etal}}. The original proof in
{\cite{Kliuchnikov-etal}} uses a direct method, i.e., without going
via the Bloch sphere representation. It is interesting to note that
Theorem~\ref{thm:completness_of_algorithm} is apparently stronger than
Corollary~\ref{cor-u2}, in the sense that the theorem obviously
implies the corollary, whereas the opposite implication is not a
priori obvious.

\begin{corollary}\label{cor-u2}
  Let $U\in U(2)$ be a unitary matrix. Then $U$ is a Clifford+$T$
  operator if and only if the matrix entries of $U$ are in the ring
  $\Dw$.
\end{corollary}

\begin{proof}
  Again, the ``only if'' direction is trivial by
  Remark~\ref{rem-necessary}.  For the ``if'' direction, it suffices
  to note that, by {\eqref{eqn-bloch-conversion}}, whenever $U$ takes
  its entries in $\Dw$, then $\hat U$ takes its entries in
  $\Ds$. Therefore, by Theorem~\ref{thm:completness_of_algorithm},
  $\hat U$ is the Bloch sphere representation of some Clifford+$T$
  operator $V$. Since $\hat V=\hat U$, $U$ and $V$ differ only by a
  phase $\phi$. Since $\phi I=UV\da$, we must have $\phi\in\Dw$, but
  this implies that $\phi=\omega^{\ell}$ for some $\ell\in\Z$, so that
  $U=\phi V$ is Clifford+$T$.
\end{proof}

%----------------------------------------------------------------------
\section{Matsumoto-Amano normal forms and {\boldmath $U(2)$}}

By Theorems~\ref{thm:ma-decomposition} and
{\ref{thm:completness_of_algorithm}}, we can efficiently convert
between a Clifford+$T$ operator $U\in U(2)$, its Bloch sphere
representation $\hat U\in SO(3)$, and its Matsumoto-Amano normal
form. Moreover, the $T$-count of the Matsumoto-Amano normal form is
exactly equal to the least denominator exponent $k$ of $\hat U$. On
the other hand, the relationship between the $T$-count and the least
denominator exponent of $U$ is more complicated.  In this section, we
establish some results directly relating the $T$-count to properties
of the matrix $U\in U(2)$. Such results can be proved by induction on
Matsumoto-Amano normal forms.

\begin{lemma}\label{lem-divisiblebyroot2}
  For all $t$ in $\Zw$, $(t+t\da)$ is divisible by $\sqrt2$.
\end{lemma}

\begin{proof}
  Let $t=a\omega^3+b\omega^2+c\omega+d$. A calculation shows that
  $t+t\da = (-d\omega^3+d\omega+c-a)\sqrt2$.
\end{proof}

\begin{definition}[Denominator exponent]
  Denominator exponents in $\Dw$ are defined similarly to those in
  $\Ds$ (cf. Definition~\ref{def-denomexp}).  Let $t\in\Dw$. A natural
  number $k\geq 0$ is called a {\em denominator exponent} for $t$ if
  $\rt{k} t \in \Zw$. The least such $k$ is called the {\em least
    denominator exponent} of $t$.
\end{definition}

\begin{definition}[Residues]
  Let $\Zbw = \Zw/(2) = \s{p\omega^3+q\omega^2+r\omega+s \mid
    p,q,r,s\in\Zb}$. Note that $\Zbw$ is a ring with exactly 16
  elements, which we call {\em residues}. We usually abbreviate a
  residue $p\omega^3+q\omega^2+r\omega+s$ by the string of binary
  digits $pqrs$.
  Consider the ring homomorphism $\rho : \Zw\to\Zbw$ defined by
  \[ \rho(a\omega^3+b\omega^2+c\omega+d) =
  \parity{a}\omega^3+\parity{b}\omega^2+\parity{c}\omega+\parity{d}.
  \]
  We call $\rho$ the {\em residue map}, and we call $\rho(t)$ the {\em
    residue} of $t$.
\end{definition}

We say that an operation $f:\Zw\to\Zw$ is {\em well-defined on
  residues} if for all $t,s$, $\rho(t)=\rho(s)$ implies
$\rho(f(t))=\rho(f(s))$. Table~\ref{tab-residue} shows several
important operations that are well-defined on residues.
\begin{table}
  \[ \begin{array}{c|c|c|c} \rho(t) & \rho(\sqrt{2}\,t) & \rho(t\da t) & \rho(\frac{t+t\da}{\sqrt2})\\\hline
    0000 & 0000 & 0000 & 0000 \\
    0001 & 1010 & 0001 & 1010 \\
    0010 & 0101 & 0001 & 0001 \\
    0011 & 1111 & 1010 & 1011 \\

    0100 & 1010 & 0001 & 0000 \\
    0101 & 0000 & 0000 & 1010 \\
    0110 & 1111 & 1010 & 0001 \\
    0111 & 0101 & 0001 & 1011 \\
  \end{array}\qquad
  \begin{array}{c|c|c|c} \rho(t) & \rho(\sqrt{2}\,t) & \rho(t\da t) & \rho(\frac{t+t\da}{\sqrt2})\\\hline
    1000 & 0101 & 0001 & 0001 \\
    1001 & 1111 & 1010 & 1011 \\
    1010 & 0000 & 0000 & 0000 \\
    1011 & 1010 & 0001 & 1010 \\

    1100 & 1111 & 1010 & 0001 \\
    1101 & 0101 & 0001 & 1011 \\
    1110 & 1010 & 0001 & 0000 \\
    1111 & 0000 & 0000 & 1010 \\
  \end{array}
  \]
  \caption{Some operations on residues}\label{tab-residue}
\end{table}

\begin{definition}[$k$-residues]
  Let $t\in\Dw$ and let $k$ be a (not necessarily least) denominator
  exponent for $t$.  The {\em $k$-residue of $t$}, in symbols
  $\rho_k(t)$, is defined to be
  \[ \rho_k(t) = \rho(\rt{k} t).
  \]
\end{definition}

\begin{remark}[Reducibility]
  We say that a residue $x\in\Zbw$ is {\em reducible} if it is of the
  form $\sqrt{2}\,y$, for some $y\in\Zbw$.  Table~\ref{tab-residue}
  shows that the reducible residues are $0000$, $0101$, $1010$, and
  $1111$.
\end{remark}

The concepts of denominator exponents, least denominator exponents,
residues, $k$-residues, and reducibility all extend in an obvious
componentwise way to vectors and matrices.

\begin{definition}
  Recall that $\sS$ is the 64-element subgroup of the Clifford group
  in $U(2)$ spanned by $S$, $X$ and $\omega$. In a way that is
  analogous to Remark~\ref{rem-right-action}, there is a well-defined
  right action of $\sS$ on the set of $2\times 2$ residue matrices,
  defined by $U\bullet A := U\cdot \rho(A)$. We write $\sim_{\sS}$ for
  the equivalence relation induced by this right action. In other
  words, for residue matrices $U,V$, we write $U\sim_{\sS}V$ if there
  exists some $A\in\sS$ such that $U\bullet A=V$. In elementary terms,
  $U\sim_{\sS}V$ holds if and only if $V$ can be obtained from $U$ by
  some combination of:
  \begin{enumerate}
  \item Shifting all of the entries in the matrix by 1,2 or 3
    positions.  This corresponds to the action of a power of $\omega$.
    \item Swapping the two columns. The corresponds to the action of $X$.
    \item Shifting the entries of the second column by two
      positions. This corresponds to the right action of $S$.
  \end{enumerate}
\end{definition}

\begin{lemma}\label{lem-columncannotbe1s}
  Let $k\geq 2$, and let
  $v=\frac{1}{\rt{k}}\begin{pmatrix}u\\t\end{pmatrix}$ be any vector with
  $u,t\in\Zw$. If $\rho(u),\rho(t)\in\s{0001,0010,0100,1000}$, then
  $v$ is not a unit vector.
\end{lemma}

\begin{proof}
  By assumption, we have $u=\omega^\jay(1+2a)$ and
  $t=\omega^\ell(1+2b)$, for some $\jay,\ell\in\Z$ and $a,b\in\Zw$.
  Suppose that $v$ is a unit vector. Then
  \begin{align*}
    2^k &= u\da u + t\da t  \\
        &= 1+2(a+a\da) +4a\da a + 1+2(b+b\da) +4b\da b,
  \end{align*}
  so $2 = 2^k - 2(a+a\da) - 2(b+b\da) - 4a\da a -4b\da b$. By
  Lemma~\ref{lem-divisiblebyroot2}, the right-hand side of this is
  divisible in $\Zw$ by $2 \sqrt{2}$, while the left-hand side is
  not. Thus we have a contradiction.
\end{proof}

\begin{lemma}\label{lem-reduction}
  Given a unitary $2\times 2$-matrix $U$ with entries in $\Dw$. Assume
  that $U$ has denominator exponent $k\ge2$, such that:
  \begin{enumerate}
    \item $\rho_{k+1}(U) \sim_{\sS} \umatrix{0101}{0101}{0101}{0101}$ and
    \item $\rho_k(H U) \sim_{\sS} \umatrix{0011}{0011}{0110}{0110}$ or
      $\rho_k(H U) \sim_{\sS} \umatrix{0011}{0011}{1001}{1001}$.
  \end{enumerate}
  Then $\rho_k(U) \sim_{\sS} \umatrix{1000}{0111}{0111}{1000}$.
\end{lemma}

\begin{proof}
  Referencing Table~\ref{tab-residue}, we see that the first condition
  limits the possible choices for the entries of $\rho_k(U)$ to the
  set $\{0010, 0111, 1000, 1101\}$. The second condition implies that
  $\rho_{k+1}(H U)$ is reducible and in fact that each entry is
  $1111$.  This means each column of $\rho_k(U)$ must be either
  $(0010, 1101)^T, (1101, 0010)^T, (0111,1000)^T$ or
  $(1000,0111)^T$. As we are considering $\sim_{\sS}$-equivalence
  classes, we can assume without loss of generality that the columns
  are in $\s{(1000, 0111)^T, (0111, 1000)^T}$.  But by
  Lemma~\ref{lem-columncannotbe1s}, we cannot have a row like $(1000,
  1000)$, and therefore
  \[ \rho_k(U) \sim_{\sS} \umatrix{1000}{0111}{0111}{1000}.\qedhere
  \]
\end{proof}

\begin{convention}
  For the purposes of the following theorem, we will consider the
  following slight variant of the Matsumoto-Amano normal form: we
  decompose the rightmost Clifford operator into up to three gates as
  $(\emptyseq\mid H\mid SH)\,\cS$, where $\cS\in\sS$. Since every
  Clifford operator can be uniquely written in this way (see
  Lemma~\ref{lem-SH}), this does not change the normal form in an
  essential way. It does, however, allow us to define the $H$-count of
  a normal form, in addition to its $T$-count. Here is the regular
  expression for the modified normal form:
  \begin{equation}\label{eqn-ma-h}
    (T\mid\emptyseq)\,(HT\mid SHT)^*\,(\emptyseq\mid H\mid SH)\,\cS.
  \end{equation}
\end{convention}

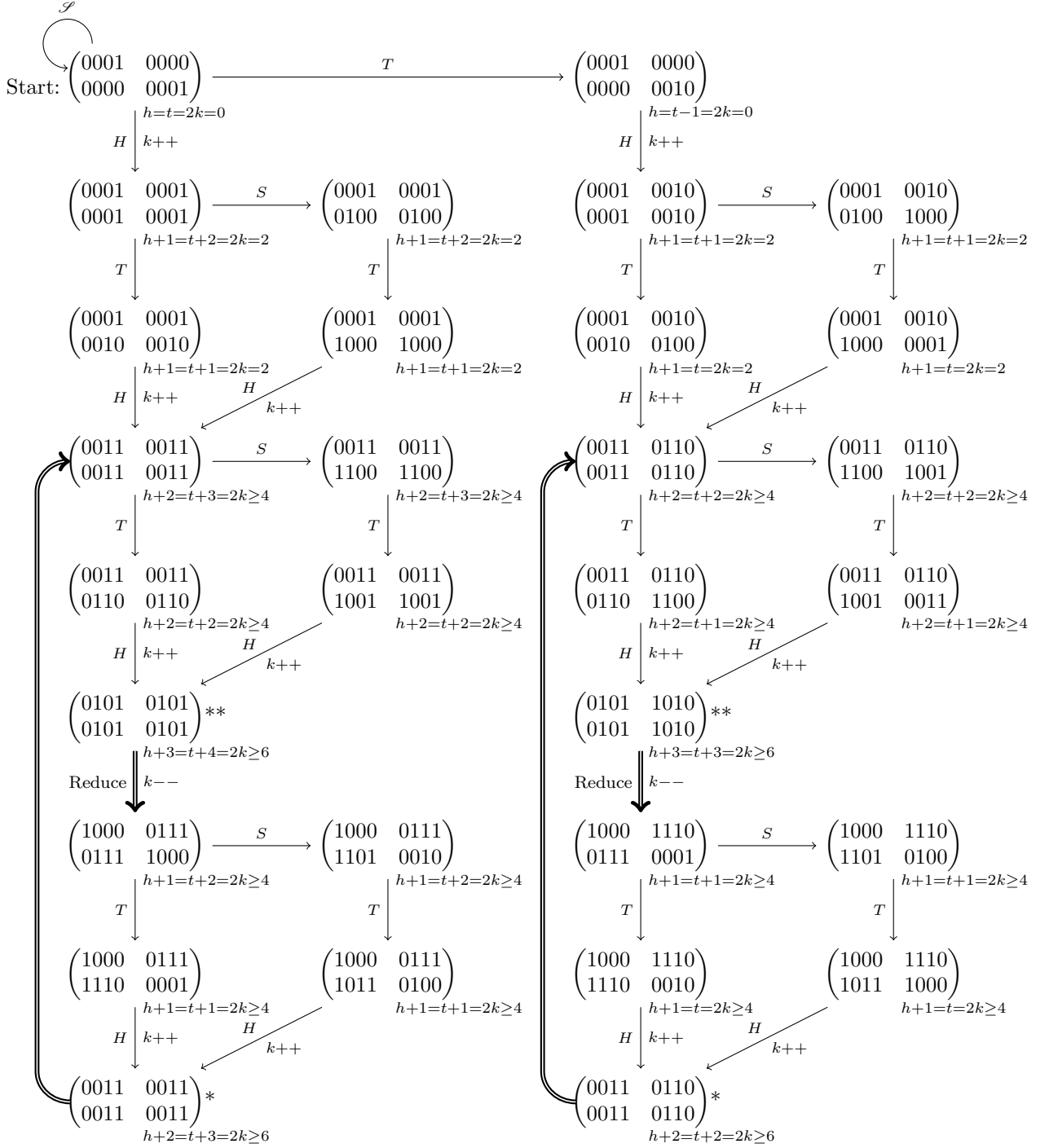
\begin{figure}
  \begin{tikzpicture}[node distance=1cm and 1.6cm]
    \def\labelstyle{\scriptstyle}
    \node (filler) {};

    \node (r1c1) [right=of filler]
          {\umatrix{0001}{0000}{0000}{0001}};

    \node [xshift=-4mm,yshift=-1.5mm] at (r1c1.west) {Start:};

    \draw [->] ([xshift=-7mm] r1c1.north) arc[start angle=0, end angle=270, x radius=4mm, y radius=4mm];
    \node at ([xshift=.15cm, yshift=.6cm] r1c1.north west) {$\labelstyle \sS$};

    \node [xshift=0cm,yshift=-.03cm,right] at (r1c1.south) {$\scriptstyle h=t=2k=0$};

    \node (r2c1) [below=of r1c1]
          {\umatrix{0001}{0001}{0001}{0001}}
          \leftrightfrom{H}{k\pp}{r1c1};
    \node [xshift=0cm,yshift=-.03cm,right] at (r2c1.south) {$\scriptstyle h+1=t+2=2k=2$};

    \node (r2c2) [right=of r2c1]
          {\umatrix{0001}{0001}{0100}{0100}}
          edge [<-] node[above] {$\labelstyle S$} (r2c1);
    \node [xshift=0cm,yshift=-.03cm,right] at (r2c2.south) {$\scriptstyle h+1=t+2=2k=2$};

    \node (r2c3) [right=of r2c2]
          {\umatrix{0001}{0010}{0001}{0010}};
    \node [xshift=0cm,yshift=-.03cm,right] at (r2c3.south) {$\scriptstyle h+1=t+1=2k=2$};

    \node (r2c4) [right=of r2c3]
          {\umatrix{0001}{0010}{0100}{1000}}
          edge [<-] node[above] {$\labelstyle S$} (r2c3);
    \node [xshift=0cm,yshift=-.03cm,right] at (r2c4.south) {$\scriptstyle h+1=t+1=2k=2$};

    \node (r1c3) [above=of r2c3]
          {\umatrix{0001}{0000}{0000}{0010}}
          edge [<-] node[above ]{$\labelstyle T$} (r1c1)
          \leftrightto{H}{k\pp}{r2c3};
    \node [xshift=0cm,yshift=-.03cm,right] at (r1c3.south) {$\scriptstyle h=t-1=2k=0$};

    \node (r3c1) [below=of r2c1]
          {\umatrix{0001}{0001}{0010}{0010}}
          edge [<-] node[left] {$\labelstyle T$} (r2c1);
    \node [xshift=0cm,yshift=-.03cm,right] at (r3c1.south) {$\scriptstyle h+1=t+1=2k=2$};

    \node (r3c2) [right=of r3c1]
          {\umatrix{0001}{0001}{1000}{1000}}
          edge [<-] node[left] {$\labelstyle T$} (r2c2);
    \node [xshift=0cm,yshift=-.03cm,right] at (r3c2.south) {$\scriptstyle h+1=t+1=2k=2$};

    \node (r3c3) [right=of r3c2]
          {\umatrix{0001}{0010}{0010}{0100}}
          edge [<-] node[left] {$\labelstyle T$} (r2c3);
    \node [xshift=0cm,yshift=-.03cm,right] at (r3c3.south) {$\scriptstyle h+1=t=2k=2$};

    \node (r3c4) [right=of r3c3]
          {\umatrix{0001}{0010}{1000}{0001}}
          edge [<-] node[left] {$\labelstyle T$} (r2c4);
    \node [xshift=0cm,yshift=-.03cm,right] at (r3c4.south) {$\scriptstyle h+1=t=2k=2$};

    \node (r4c1) [below=of r3c1]
          {\umatrix{0011}{0011}{0011}{0011}}
          \leftrightfrom{H}{k\pp}{r3c1}
          \starportfrom{H}{k\pp}{r3c2};
    \node [xshift=0cm,yshift=-.03cm,right] at (r4c1.south) {$\scriptstyle h+2=t+3=2k\ge4$};

    \node (r4c2) [right=of r4c1]
          {\umatrix{0011}{0011}{1100}{1100}}
          edge [<-] node[above] {$\labelstyle S$} (r4c1);
    \node [xshift=0cm,yshift=-.03cm,right] at (r4c2.south) {$\scriptstyle h+2=t+3=2k\ge4$};

    \node (r4c3) [right=of r4c2]
          {\umatrix{0011}{0110}{0011}{0110}}
          \leftrightfrom{H}{k\pp}{r3c3}
          \starportfrom{H}{k\pp}{r3c4};
    \node [xshift=0cm,yshift=-.03cm,right] at (r4c3.south) {$\scriptstyle h+2=t+2=2k\ge4$};

    \node (r4c4) [right=of r4c3]
          {\umatrix{0011}{0110}{1100}{1001}}
          edge [<-] node[above] {$\labelstyle S$} (r4c3);
    \node [xshift=0cm,yshift=-.03cm,right] at (r4c4.south) {$\scriptstyle h+2=t+2=2k\ge4$};

    \node (r5c1) [below=of r4c1]
          {\umatrix{0011}{0011}{0110}{0110}}
          edge [<-] node[left] {$\labelstyle T$} (r4c1);
    \node [xshift=0cm,yshift=-.03cm,right] at (r5c1.south) {$\scriptstyle h+2=t+2=2k\ge4$};

    \node (r5c2) [right=of r5c1]
          {\umatrix{0011}{0011}{1001}{1001}}
          edge [<-] node[left] {$\labelstyle T$} (r4c2);
    \node [xshift=0cm,yshift=-.03cm,right] at (r5c2.south) {$\scriptstyle h+2=t+2=2k\ge4$};

    \node (r5c3) [right=of r5c2]
          {\umatrix{0011}{0110}{0110}{1100}}
          edge [<-] node[left] {$\labelstyle T$} (r4c3);
    \node [xshift=0cm,yshift=-.03cm,right] at (r5c3.south) {$\scriptstyle h+2=t+1=2k\ge4$};

    \node (r5c4) [right=of r5c3]
          {\umatrix{0011}{0110}{1001}{0011}}
          edge [<-] node[left] {$\labelstyle T$} (r4c4);
    \node [xshift=0cm,yshift=-.03cm,right] at (r5c4.south) {$\scriptstyle h+2=t+1=2k\ge4$};

    \node (r6c1) [below=of r5c1]
          {\umatrix{0101}{0101}{0101}{0101}\makebox[0in][l]{**}}
          \leftrightfrom{H}{k\pp}{r5c1}
          \starportfrom{H}{k\pp}{r5c2};
    \node [xshift=0cm,yshift=-.03cm,right] at (r6c1.south) {$\scriptstyle h+3=t+4=2k\ge6$};

    \node (r6c3) [below=of r5c3]
          {\umatrix{0101}{1010}{0101}{1010}\makebox[0in][l]{**}}
          \leftrightfrom{H}{k\pp}{r5c3}
          \starportfrom{H}{k\pp}{r5c4};
    \node [xshift=0cm,yshift=-.03cm,right] at (r6c3.south) {$\scriptstyle h+3=t+3=2k\ge6$};

    \node (r7c1) [below=of r6c1]
          {\umatrix{1000}{0111}{0111}{1000}}
          edge [<-, thick, double] node[left] {\footnotesize Reduce}  node[right] {$\labelstyle k\mm$} (r6c1);
    \node [xshift=0cm,yshift=-.03cm,right] at (r7c1.south) {$\scriptstyle h+1=t+2=2k\ge4$};

    \node (r7c2) [right=of r7c1]
          {\umatrix{1000}{0111}{1101}{0010}}
          edge [<-] node[above] {$\labelstyle S$} (r7c1);
    \node [xshift=0cm,yshift=-.03cm,right] at (r7c2.south) {$\scriptstyle h+1=t+2=2k\ge4$};

    \node (r7c3) [right=of r7c2]
          {\umatrix{1000}{1110}{0111}{0001}}
          edge [<-, thick, double] node[left] {\footnotesize Reduce}  node[right] {$\labelstyle k\mm$} (r6c3);
    \node [xshift=0cm,yshift=-.03cm,right] at (r7c3.south) {$\scriptstyle h+1=t+1=2k\ge4$};

    \node (r7c4) [right=of r7c3]
          {\umatrix{1000}{1110}{1101}{0100}}
          edge [<-] node[above] {$\labelstyle S$} (r7c3);
    \node [xshift=0cm,yshift=-.03cm,right] at (r7c4.south) {$\scriptstyle h+1=t+1=2k\ge4$};

    \node (r8c1) [below=of r7c1]
          {\umatrix{1000}{0111}{1110}{0001}}
          edge [<-] node[left] {$\labelstyle T$} (r7c1);
    \node [xshift=0cm,yshift=-.03cm,right] at (r8c1.south) {$\scriptstyle h+1=t+1=2k\ge4$};

    \node (r8c2) [right=of r8c1]
          {\umatrix{1000}{0111}{1011}{0100}}
          edge [<-] node[left] {$\labelstyle T$} (r7c2);
    \node [xshift=0cm,yshift=-.03cm,right] at (r8c2.south) {$\scriptstyle h+1=t+1=2k\ge4$};

    \node (r8c3) [right=of r8c2]
          {\umatrix{1000}{1110}{1110}{0010}}
          edge [<-] node[left] {$\labelstyle T$} (r7c3);
    \node [xshift=0cm,yshift=-.03cm,right] at (r8c3.south) {$\scriptstyle h+1=t=2k\ge4$};

    \node (r8c4) [right=of r8c3]
          {\umatrix{1000}{1110}{1011}{1000}}
          edge [<-] node[left] {$\labelstyle T$} (r7c4);
    \node [xshift=0cm,yshift=-.03cm,right] at (r8c4.south) {$\scriptstyle h+1=t=2k\ge4$};

    \node (r9c1) [below=of r8c1]
          {\umatrix{0011}{0011}{0011}{0011}\makebox[0in][l]{*}}
          \leftrightfrom{H}{k\pp}{r8c1}
          \starportfrom{H}{k\pp}{r8c2};
    \node [xshift=0cm,yshift=-.03cm,right] at (r9c1.south) {$\scriptstyle h+2=t+3=2k\ge6$};

    \draw [->, thick, double] ([xshift=.2cm] r9c1.west) to
            [out=180,in=270] ([xshift=-1.6cm,yshift=.5cm]  r9c1) to
            ([xshift=-1.6cm,yshift=-.5cm]  r4c1) to
            [out=90,in=180] ([xshift=.2cm]  r4c1.west);

    \node (r9c3) [below=of r8c3]
          {\umatrix{0011}{0110}{0011}{0110}\makebox[0in][l]{*}}
          \leftrightfrom{H}{k\pp}{r8c3}
          \starportfrom{H}{k\pp}{r8c4};
    \node [xshift=0cm,yshift=-.03cm,right] at (r9c3.south) {$\scriptstyle h+2=t+2=2k\ge6$};

    \draw [->, thick, double] ([xshift=.2cm] r9c3.west) to
            [out=180,in=270] ([xshift=-1.6cm,yshift=.5cm]  r9c3) to
            ([xshift=-1.6cm,yshift=-.5cm]  r4c3) to
            [out=90,in=180] ([xshift=.2cm]  r4c3.west);
  \end{tikzpicture}
  \caption{The left action of Matsumoto-Amano normal forms on
    $k$-residues over $U(2)$. All matrices are written modulo the
    right action of $\sS$.}
  \label{fig:u2-transitions}
\end{figure}

\begin{theorem}
  Let $M$ be a Matsumoto-Amano normal form as in {\eqref{eqn-ma-h}},
  and let $U\in U(2)$ be the corresponding operator. Let $t$ be the
  $T$-count and $h$ the $H$-count of $M$. Let $k$ be the least
  denominator exponent of $U$, and let $R=\rho_k(U)$ be its
  $k$-residue. Then $R$ occurs (up to $\sim_{\sS}$, and excluding
  vertices labelled ``*'' or ``**'') exactly once in
  Figure~\ref{fig:u2-transitions}. Moreover, $t$, $h$, and $k$ satisfy
  the relationship indicated on the corresponding vertex in
  Figure~\ref{fig:u2-transitions}.
\end{theorem}

\begin{proof}
  By induction on the length of the Matsumoto-Amano normal form $M$.
  The technique is the same as that of Lemma~\ref{lem-ma}, although
  there are more cases.  Figure~\ref{fig:u2-transitions} shows the
  action of Matsumoto-Amano operators on residue matrices. Each vertex
  (except vertices marked ``*'' and ``**'', which we discuss below)
  represents an $\sim_{\sS}$-equivalence class of $k$-residues. Each
  arrow represents left multiplication by the relevant operator. Thus,
  each Matsumoto-Amano normal form gives rise to a unique path in the
  graph, starting from the vertex labelled ``Start''.

  The two vertices labelled ``*'' are duplicates, and were only added
  for typographical reasons. Each such vertex should be considered the
  same as the respective vertex pointed to by the double arrow.  For
  the two vertices labelled ``**'', the associated residue matrix is 
  reducible, and reduces,
  along the double arrow marked ``Reduce'', to the residue matrix
  shown immediately below it. For the matrix marked ``**'' in the left
  column, this reduction is justified by
  Lemma~\ref{lem-reduction}. For the matrix marked ``**'' in the right
  column, it can be justified by an analogous argument.

  The label $k\pp$ on an arrow indicates that the least denominator
  exponent increases by $1$, and the label $k\mm$ indicates that it
  decreases by $1$. It is then an easy case distinction to show that
  the residues, least denominator exponents, $T$-counts, and
  $H$-counts indeed behave as shown in
  Figure~\ref{fig:u2-transitions}, and that no residue occurs more
  than once. This proves the lemma.
\end{proof}

\begin{corollary}
  Let $M$ be a Matsumoto-Amano normal form as in {\eqref{eqn-ma-h}},
  and let $U\in U(2)$ be the corresponding operator. Let $t$ be the
  $T$-count and $h$ the $H$-count of $M$, and let $k$ be the least
  denominator exponent of $U$. Then we have $2k-3\leq t\leq 2k+1$ and
  $2k-2\leq h\leq 2k$. Moreover, the differences $2k-t$ and $2k-h$
  only depend on the $k$-residue of $U$.
\end{corollary}

\begin{proof}
  Immediate from Figure~\ref{fig:u2-transitions}.
\end{proof}

%----------------------------------------------------------------------
\section{Alternative normal forms}

With the exception of the left-most and right-most gates, the
Matsumoto-Amano normal form uses syllables of the form $HT$ and
$SHT$. It is of course possible to use different sets of syllables
instead. We briefly comment on a number of possible alternatives.

%----------------------------------------------------------------------
\subsection{$E$-$T$ normal form}

Consider the Clifford operator
\[ E = HS^3\omega^3 = \frac{1}{2}\zmatrix{rr}{-1+i & 1+i \\ -1+i & -1-i}.
\]
The operator $E$ serves as a convenient operator for switching between
the $X$-, $Y$-, and $Z$-bases, due to the following properties:
\[ E^3 = I,
\quad EXE\da = Y,
\quad EYE\da = Z,
\quad EZE\da = X.
\]
The operator $E$ is often convenient for calculations; for example,
every Clifford gate can be uniquely written as $E^aX^bS^c\omega^d$,
where $a\in\s{0,1,2}$, $b\in\s{0,1}$, $c\in\s{0,\ldots,4}$,
$d\in\s{0,\ldots,7}$.  On the Bloch sphere, it represents a rotation
by 120 degrees about the axis $(1,1,1)^T$:
\[ \hat E = \zmatrix{ccc}{0&0&1\\1&0&0\\0&1&0}.
\]

The operators $E$ and $E^2$ have properties analogous to $H$ and $SH$.
Specifically, if we let $\sH = \s{I,E,E^2}$ and $\sHp = \s{E,E^2}$,
then the properties of Lemma~\ref{lem-SH} are satisfied. The proofs of
Theorem~\ref{thm-ex} and Corollary~\ref{cor-efficient} only
depend on these properties, and the uniqueness proof
(Theorem~\ref{thm-ma}) also goes through without significant
changes. We therefore have:

\begin{proposition}[$E$-$T$ normal form]
  Every single-qubit Clifford+$T$ operator can be uniquely written in
  the form
  \begin{equation}\label{eqn-et}
    (T\mid\emptyseq)\,(ET\mid E^2T)^*\,{\cC}.
  \end{equation}
  Moreover, this normal form has minimal $T$-count, and there exists a
  linear-time algorithm for symbolically reducing any sequence of
  Clifford+$T$ operators to this normal form.
\end{proposition}

% ----------------------------------------------------------------------
\subsection{$T_x$-$T_y$-$T_z$ normal form}

It is plain to see that every syllable of the $E$-$T$ normal form
(except perhaps the first or last one) consists of a 45 degree
$z$-rotation, followed by a basis change that rotates either the $x$-
or $y$-axis into the $z$-position. Abstracting away from these basis
changes, the entire normal form can therefore be regarded as a
sequence of 45-degree rotations about the $x$-, $y$-, and $z$-axes.
More precisely, let us define variants of the $T$-gate that rotate
about the three different axes:
\[ \begin{array}{l}
  T_x = ETE^2, \\
  T_y = E^2TE, \\
  T_z = T.
\end{array}
\]
Using the commutativities $ET_x = T_yE$, $ET_y = T_zE$, and $ET_z =
T_xE$, it is then clear that every expression of the form
(\ref{eqn-et}) can be uniquely rewritten as a sequence of $T_x$,
$T_y$, and $T_z$ rotations, with no repeated symbol, followed by a
Clifford operator. This can be easily proved by induction, but is best
seen in an example:
\[ \begin{array}{rcl}
  TETETE^2TEC
  &=& T_zET_zET_zE^2T_zEC \\
  &\rightarrow& T_zT_xE^2T_zE^2T_zEC \\
  &\rightarrow& T_zT_xT_yE^4T_zEC \\
  &\rightarrow& T_zT_xT_yET_zEC \\
  &\rightarrow& T_zT_xT_yT_xE^2C \\
  &\rightarrow& T_zT_xT_yT_xC'. \\
\end{array}
\]
We have:
\begin{proposition}[$T_x$-$T_y$-$T_z$ normal form]
  Every single-qubit Clifford+$T$ operator can be uniquely written in
  the form
  \[  T_{r_1}T_{r_2}\ldots T_{r_n} C,
  \]
  where $n\geq 0$, $r_1,\ldots,r_n\in\s{x,y,z}$, and $r_i\neq r_{i+1}$
  for all $i\leq n-1$. We define the $T$-count of such an expression
  to be $n$; then this normal form has minimal $T$-count. Moreover,
  there exists a linear-time algorithm for symbolically reducing any
  sequence of Clifford+$T$ operators to this normal form.
\end{proposition}

The $T_x$-$T_y$-$T_z$ normal form was first considered by Gosset et
al.~{\cite[Section 4]{Gosset-etal}}. It is, in a sense, the most
``canonical'' one of the normal forms considered here; it also
explains why $T$-count is an appropriate measure of the size of a
Clifford+$T$ operator. In a physical quantum computer with error
correction, there is in general no reason to expect the $T_z$ gate to
be more privileged than the $T_x$ or $T_y$ gates; one may imagine a
quantum computer providing all three $T$-gates as primitive logical
operations.

%----------------------------------------------------------------------
\subsection{Bocharov-Svore normal forms}

Bocharov and Svore {\cite[Prop.1]{BS}} consider the following normal
form for single-qubit Clifford+$T$ circuits:
\begin{equation}\label{eqn-bs1}
  (H\mid\emptyseq)\,(TH\mid SHTH)^*\,{\cC}.
\end{equation}
This normal form is not unique; for example, $H.H$ and $I$ are two
different normal forms denoting the same operator, as are $SHTH.Z$ and
$H.SHTH$. (Here we have used a dot to delimit syllables; this is for
readability only). Recall that two regular expressions are {\em
  equivalent} if they define the same set of strings. Using laws of
regular expressions, we can equivalently rewrite (\ref{eqn-bs1}) as
\begin{equation}\label{eqn-bs2}
  ((\emptyseq\mid T\mid SHT) (HT\mid HSHT)^*\,H{\cC}) ~\mid~ \cC.
\end{equation}
Since $H\cC$ is just a redundant way to write a Clifford operator, we
can simplify it to $\cC$; moreover, in this case, $\emptyseq\cC$ and
$\cC$ are the same, so (\ref{eqn-bs2}) simplifies to
\begin{equation}\label{eqn-bs3}
    (\emptyseq\mid T\mid SHT) (HT\mid HSHT)^*\,{\cC}.
\end{equation}
Moreover, since $SHT=HSHT.X$, any expression starting with $SHT$ can be
rewritten as one starting with $HSHT$, so the $SHT$ syllable is
redundant and we can eliminate it:
\begin{equation}\label{eqn-bs4}
    (\emptyseq\mid T) (HT\mid HSHT)^*\,{\cC}.
\end{equation}
Let us say that an operator is in {\em Bocharov-Svore normal form} if
it is written in the form (\ref{eqn-bs4}). This version of the
Bocharov-Svore normal form is indeed unique; note that it is almost
the same as the Matsumoto-Amano normal form, except that the syllable
$SHT$ has been replaced by $HSHT$. Since the set $\sH=\s{I,H,HSH}$
satisfies Lemma~\ref{lem-SH}, existence, uniqueness, $T$-optimality,
and efficiency are proved in the same way as for the Matsumoto-Amano
and $E$-$T$ normal forms.

Bocharov and Svore {\cite[Prop.2]{BS}} also consider a second normal
form, which has Clifford operators on both sides, but the first four
interior syllables restricted to $TH$:
\begin{equation}
  \cC\,(\emptyseq\mid TH\mid (TH)^2 \mid (TH)^3 \mid (TH)^4(TH\mid SHTH)^*)\,\cC
\end{equation}
However, this normal form is not at all unique; for instance,
$Z.TH$ and $TH.X$ denote the same operator, as do $YS.TH.TH$ and
$TH.TH.X\omega$.

%----------------------------------------------------------------------
\section{Conclusion}

In the five years since Matsumoto and Amano published their normal
form for single-qubit Clifford+$T$ circuits, exact and approximate
synthesis of quantum circuits has only grown in importance. The
Solovay-Kitaev algorithm has been replaced by a new generation of
efficient number-theoretic approximate synthesis algorithms that
achieve circuit sizes that are linear in $\log(1/\epsilon)$
{\cite{KMM-approx,Selinger-newsynth,KMM-practical,KBS13}}. Progress
has also been made on exact synthesis, and there are now nice
algebraic characterizations of the Clifford+$T$ group, both on one
qubit {\cite{Kliuchnikov-etal}} and multiple qubits
{\cite{Giles-Selinger}}. While there are still many open questions in
the multi-qubit case, it appears that single-qubit Clifford+$T$
circuits are by now exceptionally well-understood. The Matsumoto-Amano
normal form is an important part of this understanding. We hope that
with this paper, we have fleshed out the basic properties of this
remarkable normal form, and contributed to making it more widely
known.

%----------------------------------------------------------------------
\section{Acknowledgements}

Thanks to Xiaoning Bian for reporting typos.

%----------------------------------------------------------------------
\bibliographystyle{abbrv}
\bibliography{ma-remarks}

%----------------------------------------------------------------------
\end{document}